\documentclass[11pt,a4paper,fleqn,reqno,twoside]{article}
\usepackage{amsmath,amssymb,amsthm,latexsym,times}
\usepackage{natbib}
\usepackage[lmargin=1.in,rmargin=1.in,tmargin=1.4in,bmargin=1.3in]{geometry}
\numberwithin{equation}{section}
 \numberwithin{table}{section}

\usepackage{fancyhdr}
\makeatletter
\renewcommand\section{\@startsection {section}{1} {0pt}{20pt}{10pt} {\normalfont\large\bfseries}}
\makeatother

\makeatletter
\renewcommand\subsection{\@startsection {subsection}{1} {0pt}{20pt}{10pt} {\normalfont\large\bfseries}}
\makeatother

\renewenvironment{proof}{\rm \trivlist  \item[\hskip \labelsep{\sc Proof.}]}{\squareforqed \endtrivlist}
\newcommand\squareforqed{\hbox{\rlap{$\sqcap$}$\sqcup$}}

\theoremstyle{remark}
\newtheorem{theorem}{Theorem}
\newtheorem{lemma}{Lemma}
\newtheorem{definition}{Definition}

\theoremstyle{remark}
\newtheorem{example}{Example}

\newcommand{\HRule}{\rule{\linewidth}{0.1mm}}

\newcommand{\design}{{\mathcal D}}
\newcommand{\fraction}{\mathcal F}
\newcommand{\est}{\textrm{Est}}

\def\cocoa{{\hbox{\rm C\kern-.13em o\kern-.07em C\kern-.13em o\kern-.15em A}}}
\begin{document}
\thispagestyle{empty}

\begin{center}
{\Large {Two polynomial representations of experimental design}

\vskip 1 cm
 {\large Roberto Notari,} {\small \emph{Dipartimento di  Matematica, Politecnico di Torino \\ Turin, Italy
(\underline{roberto.notari@polito.it})}}

\bigskip
 {\large Eva Riccomagno,} {\small \emph{Dipartimento di  Matematica, Universit\`a di Genova \\  Genoa, Italy
(\underline{riccomagno@dima.unige.it})}}

\bigskip
 {\large Maria-Piera Rogantin,} {\small \emph{Dipartimento di  Matematica, Universit\`a di Genova \\ Genoa, Italy
(\underline{rogantin@dima.unige.it})}}
}
\end{center}

\noindent\HRule

{\footnotesize
 \noindent \textbf{Abstract}

{\medskip \noindent  In the context of algebraic statistics an experimental design is described by a set of polynomials
called the design ideal. This, in turn, is generated by finite sets of polynomials. Two types of generating sets are
mostly used in the literature: Gr\"obner bases and indicator functions. We briefly describe them both, how they are
used in the analysis and planning of a design and how to switch between them. Examples include fractions of full
factorial designs and designs for mixture experiments.}


}

{\small
\section{Introduction}

In the algebraic statistics literature two types of polynomial representations of an experimental design are studied:
the Gr\"obner type (see \cite{pistone|wynn:96}, \cite{pistone|riccomagno|wynn:2001}) and the indicator function type
(see \cite{fontana|pistone|rogantin:2000},  \cite{ye:2003}, \cite{pistone|rogantin:2007-b}). In this paper we compare
them, describe how to derive them from the design points and how to use them in the analysis of the design properties
by unifying and completing results from the literature.  Mainly we provide an original and efficient algorithm to
switch between the two representations. The diagram below summarizes the paper.
\begin{equation*}  \begin{array}{ccc}
    \text{Points coordinates} & & \text{Generating set} \\
    \downarrow         & \searrow \hskip-11pt \swarrow & \downarrow \\
    \text{Indicator function} & {\rightleftarrows} & \text{Gr\"obner representation} \\
  \end{array}
  \end{equation*}
 In Section \ref{sec:2repr} the two representations are described and their
relative practical advantages are discussed. Algebraic algorithms to move along the four down-arrows of the diagram are
discussed. In Section \ref{sec:changing} a theorem and an algorithm to change representation are given which do not
require the knowledge of the coordinates of the points. This is represented by the horizontal arrows in the diagram.

The horizontal arrows are particularly important in the planning stage of the experiment.  This is because designs with
a given confounding structure can be easily defined through generating sets and actual point coordinates are unknown
until the corresponding system of equations is solved. The actual number of points in the design can be computed from
the design ideal using the Hilbert function, as we do in Section~\ref{sec:algorithm}. Analogue theorem and algorithm
for mixture experiments are presented in Section \ref{sec:algorithm}. An implementation of the algorithms in the
general-purpose mathematics software package Maple is provided in the Appendix. In Section \ref{sec:ex} a large design
for a screening experiment from the chemical literature is studied.

We use the dedicated symbolic softwares {\cocoa}, see \cite{CocoaSystem}, and the general purpose software Maple,
see \cite{maple:91}. The provided algorithms can be easily implemented in other softwares.

\begin{example} \label{running_example}
To illustrate the main points of our discussion we use the two simple designs, $\fraction_{A}$ and $\fraction_{P}$
below
\begin{eqnarray*}
 \fraction_A & = & \left\{ (1,0),(-1,0), (0,1), (0,-1)\right\} \\
 \fraction_P & = & \left\{  (1,0,0), (0,1,0), (0,0,1), (1/3,1/3,1/3),  (1/2,1/2,0), (1/2,0,1/2), (0,1/2,1/2)\right\}
\end{eqnarray*}
 Note that the components of each point of $\fraction_{P}$ sum to one.
\hfill $\square$
\end{example}

\section{The two representations}\label{sec:2repr}

We begin with some unavoidable algebraic notions. Relevant references to  polynomial algebra can be found in the
textbooks by \cite{cox|little|oshea:97,cox|little|oshea:2005} and \cite{kreuzer|robbiano:2000,kreuzer|robbiano:2005}.

Let $k$  be a computable numerical field, and $ k^m $ be the affine $m-$dimensional space. We consider a design with $m$ factors,
where the levels of each factor are coded with integer, rational, real or complex numbers. In practical situations $k$
is the set of the rational numbers $\mathbb Q$. For the indicator function representation, we need an extension of
$\mathbb Q$ to include the imaginary unit and some irrational real numbers. This is a computable set. Then, a design
$\fraction$ is a finite set of $n$ distinct points in $ k^m$. Let $R=k[x_1,\ldots,x_m]$ be the polynomial ring in $m $
indeterminates with coefficients in $k$. The indeterminates in $R$ correspond to the design factors.

 Three notions from algebraic geometry/commutative algebra are corner stones.

\begin{itemize}
\item \emph{Ideal of a design.} The \emph{design ideal} of $ \fraction$ is
\begin{equation*}
I(\fraction) = \{ f \in R \vert f(\zeta) = 0 \mbox{ for all } \zeta \in \fraction \}.
\end{equation*}
$ I(\fraction) $ is an ideal, i.e. $ f+g\in I(\fraction) $
for all $ f,g \in I(\fraction) $ and $fg\in I(\fraction) $ for every $f\in I(\fraction) $ and $ g\in R.$ The Hilbert
Basis theorem states that every polynomial ideal is finitely generated. Thus, there exist $f_1,\dots,f_r \in
I(\fraction) $ such that
\begin{equation*}
f \in I(\fraction) \mbox{ if, and ony if, } f = \sum_{i=1}^r s_i f_i \mbox{ for some } s_i\in R.
\end{equation*}
The set of generators $f_1,\ldots,f_r$ is not unique. The ideal generated by $f_1,\ldots,f_r$ is indicated by $\langle
f_1,\ldots,f_r \rangle$. Conversely, given an ideal $ I,$ the set of zeros of $I$ is, by definition, an algebraic variety
and corresponds to the zero-set of any of its generator sets.

\item \emph{Interpolation.} For any $ k-$valued function, $F$, defined on a design $\fraction$,  there exist (interpolating)
polynomials $f\in R$ such that $f(\zeta)=F(\zeta)$ for all $\zeta\in \fraction$.

\item \emph{Quotient space.} A standard algebraic construction is the quotient ring $ R/J $ for any ideal $ J \subseteq
R.$ The relation $ \sim $ defined as $ \{ f \sim g \mbox{ if, and only if, } f - g \in J \} $ is an equivalence
relation. The elements of $ R/J $ are the equivalence classes of $ \sim.$ $ R/J $ inherits a ring structure from $ R $
by defining sum and product of classes as $ [f] + [g] = [f+g], [f] [g] = [f g].$

If $f, \, g \in R $ interpolate the same function $ F $ defined on the design $ \fraction,$ that is to say, $ f$ and
$g$ are aliased on $ \fraction $, then $ f - g $ is zero if evaluated on each $ \zeta \in \fraction,$ and so $ f - g
\in I(\fraction).$ Hence, there exists a unique class in $ R/I(\fraction) $ that contains all the polynomials
interpolating the same function $F$.
\end{itemize}

In algebraic geometry a design $\fraction$ is seen as a zero-dimensional variety. The focus both in algebraic
statistics and in this paper switches from the design $\fraction$ to its ideal $ I(\fraction)$. As we shall see below,
the Gr\"obner representation and the indicator function representation of $\fraction$ are nothing else than two sets of
generators of $I(\fraction)$. Replicated points can be considered but some technical issues, which are briefly
illustrated in Example~\ref{examplereplicates}, occur which are outside the scope of this paper.

\begin{example}\label{examplereplicates}
In $\mathbb Q[x_1,x_2]$ consider the two ideals $I_1$ and $I_2$ defined as $I_1=\langle x_1, x_2^2 \rangle$ and
$I_2=\langle   x_1 + x_2,x_2^2 \rangle$. The zero sets of $I_1$ and $I_2$ are equal and consist of the point $(0,0)$
with multiplicity two, as can be easily checked by solving the two systems of equations $x_1=x_2^2=0$ and $x_1 + x_2 =
x_2^2 = 0$. But the two ideals are not equal because the polynomial $x_1$ is in $I_1$ but not in $I_2$ and conversely $
x_1 + x_2 \in I_2 $ but not in $ I_1$.

In general, questions like equality of ideals ($I_1=I_2$), membership of a polynomial to an ideal ($x_1\in I_1$ and
$x_1\not\in I_2$), intersection and sum of ideals can be handled by using computer algebra softwares.

\hfill $\square$
\end{example}

\subsection{Indicator function} \label{sec:if}
To define the indicator function of $\fraction$ we must consider $\fraction$ as a subset of a larger design $\mathcal
D\subset k^m$. Usually, but not necessarily, $\mathcal D$ has the structure of a full factorial design.
 The indicator function $F$ of $\fraction\subset \mathcal D$ is the response function
\begin{equation}\label{eq:F}
F(\zeta)=\left\{ \begin{array}{ll}
1 & \text{if } \zeta \in \fraction \\
0 & \text{if } \zeta \in \mathcal D\setminus \fraction .
\end{array} \right.
\end{equation}

The polynomial indicator function for two level fractional factorial designs were introduced in
\cite{fontana|pistone|rogantin:97} and \cite{fontana|pistone|rogantin:2000} and independently in \cite{tang|deng:99} with
a slightly different presentation. An extension to two-level designs with replicates is in \cite{ye:2003}
and to multilevel factors, using orthogonal polynomials with integer coding, in \cite{cheng|ye:2004}.

The case of  factorial designs is treated in \cite{pistone|rogantin:2007-b}, where the $n_j$ levels of each factor are
coded by the $n_j$-th roots of the unity, $j=1,\dots,m$. With this coding an orthonormal base of the response space on
the design is formed by the set of all the monomial terms:
 \begin{equation*}
 \left\{ x^\alpha  ,  \  \alpha \in L \right\} \qquad \textrm{ and } \
 L = \left\{   \alpha=(\alpha_1,\ldots,\alpha_m) \  , \ \alpha_j = 0,\ldots,n_j-1 \  \textrm{ and } \  j=1,\ldots,m\right\} \ .
 \end{equation*}
The indicator function $F$ is a real valued polynomial with complex coefficients: $\sum_{\alpha \in L} b_\alpha \
X^\alpha(\zeta)$, $ \zeta\in\design$.
 In this case, the coefficients are related to many
interesting properties of the fraction in a simple way: orthogonality among the factors and interactions, projectivity,
aberration and regularity. For instance, the fraction is regular if and only if all the coefficients are equal to the
ratio between the number of fraction points and the number of the full design points; the level of a simple factor of an
interaction occurs equally often in the fraction if and only if the coefficient of the corresponding term is zero; two
simple factors or interactions are orthogonal if and only if the coefficient of the term with exponent the sum of the
two exponents is zero; a fraction is an orthogonal array of strength $s$ if and only if the coefficients of the terms
of order lower than $s$ are zero.

\begin{example}
The fraction of a $3^4$ full factorial design $
 \fraction_R = \{(1,1,1,1), (1,\omega_1,\omega_1,\omega_1),(1,\omega_2,\omega_2,\omega_2),$ \\$
 (\omega_1,1,\omega_1,\omega_2),(\omega_1,\omega_1,\omega_2,1),
 (\omega_1,\omega_2,1,\omega_1), (\omega_2,1,\omega_2,\omega_1),(\omega_2,\omega_1,1,\omega_2),
 (\omega_2,\omega_2,\omega_1,1) \}
$, where $1,\omega_1,\omega_2$ are the cubic roots of the unity, is a regular fraction; in fact, its indicator function
is
 \begin{equation*} F =\frac 1 9\left( 1 + x_2 x_3 x_4 + x_2^2 x_3^2 x_4^2 + x_1 x_2 x_3^2  + x_1^2 x_2^2 x_3  + x_1
x_2^2 x_4 + x_1^2 x_2 x_4^2 + x_1 x_3 x_4^2 + x_1^2 x_3^2 x_4 \right).
 \end{equation*}
The fraction points are $1/9$ of the $3^4$ design; in fact the constant term is $1/9$. Moreover, each factor is
orthogonal to the constant term as shown by the fact that the coefficients of the terms of order 1 are 0. Any two
factors are mutually orthogonal; in fact the coefficients of the terms of order 2 are 0. The interaction terms
appearing in the indicator function are the ``defining words'' of the regular fraction.
 \end{example}
\begin{example}
The fraction of a $2^5$ full factorial design $ \fraction_O= \{(1, 1, 1, 1, 1),(  1, 1, 1,-1, 1),(  1, 1,-1,-1, 1)$,
\\$( 1,-1,-1, 1, 1), ( -1, 1,-1, 1, 1), ( -1,-1, 1, 1, 1), ( -1,-1, 1,-1, 1),( -1,-1,-1,-1, 1),( 1,1,-1,-1,-1),\\(
1,-1, 1, 1,-1),(  1,-1, 1,-1,-1), (  1,-1,-1, 1,-1),( -1, 1, 1, 1,-1),( -1, 1, 1,-1,-1),( -1, 1,-1, 1,-1),\\(
-1,-1,-1,-1,-1)\}$ is an orthogonal array of strength 2; in fact, its indicator function \begin{equation*} F = \frac 1
2  - \frac 1 4x_1x_2x_4 + \frac 1 4x_1x_2x_5 + \frac 1 4x_1x_2x_3x_4 + \frac 1 4x_1x_2x_3x_5 \end{equation*}  contains
only terms of order greater than 2, together with the constant term. \hfill $\square$
\end{example}

When the coordinates of the points in $\fraction$ and $\mathcal D$ are known, the indicator function $F$ can be
computed using some form of interpolation formula for Equation (\ref{eq:F}). The
{\cocoa} function \texttt{IdealAndSeparatorsOfPoints} is used in Example~\ref{running_example2}. If the complex coding is used, the coefficients of the
indicator function of a fraction of a full factorial design can be easily computed from the sum of the values of each
monomial response on all the fraction points: $b_\alpha=\frac 1 {\# \design} \ \sum_{\zeta \in \fraction} x^\alpha
(\zeta)$.

\begin{example}
We consider the fraction $\fraction = \{(-1,-1,1), (-1,1,-1)\}$ of the $2^3$ full factorial design. All the monomial
responses  on $\fraction$ are
\begin{equation*}
\begin{array}{r|r|r|r|r|r|r|r}
 1 & x_1 & x_2 & x_3 & x_1 x_2 & x_1 x_3 & x_2 x_3 & x_1 x_2 x_2 \\ \hline
 1 &  -1 &  -1 &   1 &       1 &      -1 &      -1 &       1  \\
 1 &  -1 &   1 &  -1 &      -1 &       1 &      -1 &       1
\end{array}
\end{equation*}
and the coefficients $b_\alpha$ are:
\begin{equation*}
 b_{(0,1,0)}= b_{(0,0,1)}= b_{(1,1,0)}=b_{(1,0,1)}=0  \qquad b_{(0,0,0)}=b_{(1,1,1)}= \frac 2 4 \qquad  b_{(1,0,0)}= b_{(0,1,1)}= - \frac 2 4 \ .
\end{equation*}
Hence, the indicator function is $ F= \frac 1 2 \left( 1 - x_1 -x_2 x_3+x_1x_2 x_3\right)$.
\end{example}

\begin{example}[Continuation of Example~\ref{running_example}]\label{running_example2}
 The indicator function of $\fraction_{A}$ as subset of a $3^2$ factorial design is \begin{equation*} F = -2\, x_1 x_2 +x_1^2+x_2^2. \end{equation*} A {\cocoa} algorithm for the computation is
provided in Item 1 of the Appendix. If the levels are coded with the 3-rd roots of the unity, the indicator function is
\begin{equation*} F = \frac 4 9 (1+x_1+x_2-x_1^2-x_2^2-x_1x_2-x_1x_2^2-x_1^2x_2-x_1^2x_2^2). \end{equation*} We can check that
there are no mutually orthogonal terms. The indicator function of $\fraction_{P}$ will be computed in Example
\ref{simpleMixture1.2}.
\end{example}

\subsection{Gr\"obner bases}
When working with polynomial ideals, it is useful to choose a standard form for writing the polynomials. This can be
done by choosing a term ordering. That is an order relation on the monomials of $R$, compatible with the product of
monomials. In more details, a monomial is written as  $x^\alpha=x_1^{\alpha_1} \ldots x_m^{\alpha_m}$ with
$\alpha=(\alpha_1,\ldots,\alpha_m)$ and $\alpha_i\in \mathbb Z_{\geq 0}$ for all $i=1,\ldots,m$. The ordering relation
$\succ$ is a term ordering if 1) $x^\alpha\succ 1$ for all exponents $\alpha$ and 2) if $x^\alpha \succ x^\beta$ then
$x^{\alpha+\gamma}\succ x^{\beta+\gamma}$ for all $\alpha,\beta,\gamma\in\mathbb Z_{\geq 0}^m$. The leading term of
$f\in R$ with respect to $\succ$ is the largest term of $f$ with respect to $\succ$ and we write
$\operatorname{LT}_\succ(f),$ or $ \operatorname{LT}(f) $ if no confusion arises.

\begin{example}  \label{ex:lex}
The lexicographic term ordering is defined as $ x_1^{\alpha_1} \cdots x_m^{\alpha_m} \succ x_1^{\beta_1} \cdots
x_m^{\beta^m} $ if $ \alpha_1 = \beta_1, \dots, \alpha_{i-1} = \beta_{i-1}, \alpha_i > \beta_i$, for some $i \in
\{1,2,\dots,m\}$. In the tdeg ordering $ x_1^{\alpha_1} \cdots x_m^{\alpha_m} \succ x_1^{\beta_1} \cdots x_m^{\beta^m}
$ if, and only if, $\sum \alpha_i>\sum \beta_i$ or  $\sum \alpha_i=\sum \beta_i$ and the right-most nonzero entry of
$(\alpha_1-\beta_1,\ldots, \alpha_m-\beta_m)$ is negative. \hfill $\square$
\end{example}

Given a term ordering $\succ$ and an ideal $I\subset R$, let $\operatorname{LT}_\succ(I)=\langle
\operatorname{LT}_\succ(f):f\in I\rangle$ be the set of leading terms of all polynomials in $I$.
\begin{definition} \label{GBasis}
 Let $I$ be an ideal, $\succ$ a term  ordering and $G=\{g_1,\ldots, g_t \} \subseteq I$.
\begin{enumerate}
  \item $G$ is a \emph{Gr\"obner basis} (sometimes called a standard basis) of $I$ if $\operatorname{LT}_\succ(I)$ is
  generated by $\langle \operatorname{LT}_\succ(g):g\in G\rangle$.
  \item $G$ is a \emph{reduced} Gr\"obner basis if for all $g\in G$ the coefficient of the leading term of $g$ is $1$
  and no term of $g$ lies in $\langle \operatorname{LT}_\succ(G\setminus \{g\}) \rangle$.
  \end{enumerate}
\end{definition}
Note that a Gr\"obner basis of an ideal $I$ is a particular generator set of $I$. For every ideal $I$ and term ordering
$\succ$ there exist Gr\"obner bases of $I$ and a unique reduced Gr\"obner basis, see \cite[Ch.2]{cox|little|oshea:97}.
Gr\"obner bases of $I$ can be computed from any generator set of $I$ with the Buchberger algorithm which is implemented
in most softwares for algebraic computation. For every ideal there is a finite number of reduced Gr\"obner bases as the
termorderin varies, see \cite{mora|robbiano:88}.

\begin{example}[Continuation of Example~\ref{running_example}] \label{running_example3}
For any term ordering for which $x_1\succ x_2$ the reduced Gr\"obner basis representation of
$I(\fraction_{A})$ is given by the three polynomials $g_1=x_1^2+x_2^2-1$, $g_2=x_2^3-x_2$ and
$g_3=x_1x_2$. The polynomial $g_1$ indicates that the points of $\fraction_{A}$ are on the unit circle,
$g_2$ that the factor corresponding to $x_2$ has three levels $0,\pm 1$ and $g_3$ that at least one coordinate of each
point in $\fraction_{A}$ is zero.

The reduced Gr\"obner basis of $I(\fraction_{P})$ for any term ordering such that $x_1\succ x_2 \succ
x_3$ has five elements
\begin{equation*} \begin{array}{l}
h_1=x_1+x_2+x_3-1 \\
h_2=x_3(x_3-1/2)(x_2+x_3/2-1/2)\\
h_3=x_3 (x_2^2+x_3^2/2-x_2/2-3/4 x_3+1/4)\\
h_4=(x_2-x_3)(x_2^2+x_2x_3+x_3^2-3/2x_3-3/2x_3+1/2)\\
h_5=x_3 (x_3-1/3) (x_3-1/2) (x_3-1)
\end{array} \end{equation*}
In $h_1$ we can recognize the sum to one condition for a mixture design and in $h_5$ the levels of the $x_3$ factor.
\hfill $\square$
\end{example}

If we fix a Gr\"obner basis of an ideal $ I \subseteq R,$ then for every equivalence class $ [f] \in R/I $ there exists
a unique $ f' \in [f] $ written as combination of monomials not divisible by any monomial in $ \operatorname{LT}(I)$.
The polynomial $f'$ is called the \emph{normal form} of $ f$ and we write $NF(f)$, see \cite{cox|little|oshea:97}.
Hence, Gr\"obner bases give a tool to effectively perform sum and products in the quotient ring $ R/I$.

Given a design $ \fraction,$ the quotient ring $ R/I(\fraction) $ is a vector space of dimension equal to the
cardinality of $ \fraction$. A monomial basis of $R/I(\fraction)$ can be used as support for a statistical (saturated)
regression model as the corresponding information matrix is invertible. A vector space basis of $R/ I(\fraction)$ can
be determined by using Gr\"obner bases, and the procedure, which we call Gbasis/LT, is the following. The monomials
which are not in $ \operatorname{LT}(I(\fraction)) $ are linearly independent over the design. Call this set $
\est_\mathcal F$. They are those monomials which are not divided by any of $\operatorname{LT}(g)$ for all $g$ in a
Gr\"obner basis of $I(\fraction)$. This is equivalent to the fact that the columns of the matrix $X=\left[ \zeta^\alpha
\right]_{\zeta\in \fraction,\alpha\in L}$ are linearly independent  (\cite{pistone|riccomagno|wynn:2001}), where $ L $
is the set of the exponents of the elements of $ \est_\fraction.$

\begin{example}[Continuation of Example~\ref{running_example}]
In the setting and notation of Example \ref{running_example3}, the leading terms of the Gr\"obner basis elements of
$I(\fraction_{A})$ are $\operatorname{LT}(g_1)=x_1^2$, $\operatorname{LT}(g_2)=x_2^3$ and
$\operatorname{LT}(g_3)=x_1x_2$. The four monomials $1,x_1,x_2,x_2^2$ are not divisible by these leading terms,
equivalently the first four columns of $X$ below give an invertible matrix.
\begin{equation*}
 X=\left[ \begin{array}{rrrrr|l} 1&x_1&x_2&x_2^2& x_1^2&\zeta \\ \hline
 1& 1 & 0 & 0 & 1 & (1,0) \\
 1& 0 & 1 & 1 & 0 & (0,1) \\
 1&-1 & 0 & 0 & 1 & (-1,0) \\
 1& 0 & -1& 1 & 0 & (0,-1)
\end{array}  \right]
\end{equation*}
The linear response model build on any combination of the first four columns of $X$ is identified. The last column lists
the design points. From $g_1=x_1^2+x_2^2-1$ we deduce that $x_1^2=1-x_2^2$, that is the fifth column of $X$ is the
difference between the first column and the fourth column.

For $\fraction_{P}$, we have $ \est_{\fraction_{P}}=\{1,x_3,x_3^2,x_3^3,x_2,x_2x_3,x_2^2\}$. Note that there is no term involving $x_1$ as $g_1=x_1+x_2+x_3-1$ confounds $x_1$ with $x_2$ and $x_3$ (see Section~\ref{mixsec}). \hfill $\square$
\end{example}

The design ideal embeds all possible aliasing relations imposed on polynomial responses by a design. A Gr\"obner basis is
a special finite set of aliasing relations among polynomial responses defined on the fraction and are a basis of all
other alias relations, see \cite{holliday|riccomagno|wynn|pistone:99}. Other special finite sets can be found using the
indicator function. Theorem 4 and Example 10 in  \cite{pistone|riccomagno|rogantin:2007} present an algorithm based on
the computation of the normal form, with respect to the full design, of all the monomial responses multiplied by the
indicator function of the fraction. Knowledge of the problem to be modelled indicates whether the sets of aliasing
relations from  the indicator function or from the Gr\"obner bases are more informative.


\subsection{Designs for experiments with mixtures} \label{mixsec}
We need to refer here a short summary of \cite{maruri-aguilar|notari|riccomagno:2007}. Each point
$\zeta=(\zeta_1,\ldots,\zeta_m)\in k^m$ of a design $\fraction$ for a mixture experiment satisfies the conditions that
$\zeta_i\geq 0$ for all $i=1,\ldots,m$ and $\sum_i \zeta_i=1$. The polynomial $\sum_i x_i-1\in I(\fraction)$ and thus
not all the linear terms can be in the support of a regression model simultaneously. In particular the Gbasis/LT
procedure applied to a design for a mixture experiment returns slack models which include the identity/intercept and
miss completely one factor. For a mixture design $\fraction\subset k^m$ there exists, well defined, a unique cone
passing  through $\fraction$ and the origin:
\begin{equation*}
\mathcal C_{\fraction} = \left\{ a\zeta : \zeta\in \fraction \text{ and }
  a\in \mathbb R \right \} \subseteq k^m  .
\end{equation*}
This can be thought of as a projective variety. The Gbasis/LT procedure is specialized to mixture designs exploiting
the fact that projective varieties and homogeneous polynomials are naturally associated.  Consider a term order, the
cone $\mathcal C_{\fraction}$ and all homogeneous polynomials of degree $s$ in $R$. Compute a Gr\"obner basis of
$\mathcal C_{\fraction}$, its leading terms and the set of monomials of degree $s$ not divisible by  the leading terms.
Then, the information matrix for this set and $\fraction$ is invertible  (see
\cite{maruri-aguilar|notari|riccomagno:2007}).

\begin{example}
Consider the design $\fraction=\{(0,0,1),(0,1,0),(1,0,0),(1/3,1/3,1/3)\}\subset\fraction_P$,  and any term ordering
such that $x_1\succ x_2\succ x_3$. Then $I({\mathcal
  C}_{\fraction}) = \langle \underline{x_1x_3}-x_2x_3,  \underline{x_1x_2}-x_2x_3,
\underline{x_2^2x_3}-x_2x_3^2 \rangle$. The leading terms are underlined. In Table 2.1 various homogeneous models
identified by $\fraction$ are given. Notice that they are Kronecker models generalizing those in
\cite{draper|pukelsheim:98}.
\renewcommand{\arraystretch}{1.1}

\begin{table}[h]\label{aaaa}
 \caption{\small Homogeneous models identified by $\fraction=\{(0,0,1),(0,1,0),(1,0,0),(1/3,1/3,1/3)\}$ }
\begin{center}
{\small\begin{tabular}{lll}
 $s$&\text{list of monomials of degree $s$}&\text{degree $s$ standard monomials} \\ \hline
 0 & 1 & 1\\
 1 & $x_1,\, x_2, \, x_3$ & $x_1,\, x_2, \,x_3$\\
 2 & $x_1^2,\, x_1x_2, \, x_2^2, \, x_1x_3, \,x_2x_3, \, x_3^2$ & $x_1^2, \,x_2^2,\,x_2x_3, \, x_3^2$\\
 3 & $x_1^3,\,  x_1^2x_2,\,  x_1x_2^2,\,  x_2^3,\,  x_1^2x_3, x_1x_2x_3,\,  x_2^2x_3,\,  x_1x_3^2,\,  x_2x_3^2,\,  x_3^3$
 & $x_1^3,\,x_2^3,\,  x_2x_3^2,\,  x_3^3$\\
 $s>3$ &$x_1^s,\, x_1^{s-1}x_2,\, x_1^{s-2}x_2^2,\, \ldots,\, x_3^s$&$x_1^s,\, x_2^s,\, x_2x_3^{s-1},\, x_3^s$
\end{tabular}
\\
\hfill $\square$ } \end{center} \end{table}
\end{example}

\begin{example} [Continuation of Example~\ref{running_example}]
\label{running_example4} The set $\{ \underline{x_2^2 x_3}-x_2 x_3^2, \underline{x_1^2 x_3} - x_1 x_3^2,
\underline{x_1^2 x_2} -x_1 x_2^2 \}$ is the reduced Gr\"obner basis of $I(\mathcal C_{\fraction_{P}})$ with respect to
the  tdeg ordering with $x_1\succ x_2\succ x_3$. For $s=3$, $ \est_{\fraction_P}=\{x^3,xy^2,y^3, $ $ xyz,xz^2,yz^2,z^3 \}$
gives the support for a homogeneous saturated regression model identified by $\mathcal C_{\fraction_{P}}$.
\end{example}

We need to observe now that ratios of homogeneous polynomials of the same degree are functions well defined on the
affine cone of a mixture design.
\begin{example} \label{ex:lines}
Consider the three points $ P_1=(1,0,0), P_2=(0,1,0), P_3=(0,0,1) $ and the lines $L_i $ through $ P_i $ and the origin $
(0,0,0)$. The cone over the points $ P_1, P_2, P_3 $ is equal to $ L_1 \cup L_2 \cup L_3.$ Let $ F $ be the function
which assumes the value $ i $ on $ L_i, i=1, 2, 3$. First, we show that $ F $ cannot be represented as a polynomial.
In fact, if $ f $ is a polynomial such that $ f(x,0,0) = 1 $ for each $ x \not= 0,$ then $ f = 1 + f_1(y,z).$ But, $
f(0,y,0) = 1 + f_1(y,0) = 2 $ for every $ y \not= 0,$ and so $ f = 2 + f_2(z).$ Hence, $ f(x,0,0) = 2,$ and so $ F $
cannot be represented by a polynomial. Next, note that $ \frac{x+2y+3z}{x+y+z} $ represents $ F $ on the considered
cone. \hfill $ \square $
\end{example}

The above leads to the following definition, which specializes the ideal of indicator function to mixture designs. The
larger design $\design$ could be any mixture design, i.g. a simple lattice, see \cite{scheffe:58} or a simple centroid
design, see \cite{scheffe:63}. Example \ref{ex:lines} shows that we need to consider ratios of polynomials of the same
degree to define a function on $\mathcal C_\design$ and not simply a polynomial, see Definition 2.B2 and 2.B1 below. Furthermore, the
notion of separator as introduced is for consistency with algebraic standard, and takes zero value in $\design
\setminus \fraction$.

\begin{definition} Let $\fraction\subseteq \mathcal D \subset k^m$ be designs for a mixture experiment.
\begin{itemize}
\item [A1)] A \emph{separator} of $\zeta\in \fraction$ is any homogeneous polynomial $S_\zeta$ such that $S_\zeta\not\in
  I(\mathcal C_{\{F\}})$ and $S_\zeta\in I(\mathcal C_{\mathcal  F\setminus \{\zeta\} })$.
\item [A2)] The \emph{separator function} of $\zeta\in \fraction$ is
  $S_{\{\zeta\}}= \displaystyle \frac{S_\zeta}{\left( \sum_{i=1}^m x_i\right)^{s_\zeta}} $
  where $s_\zeta$ is the degree of $S_\zeta$.
\item [B1)] A \emph{separator} of $\fraction \subset \mathcal D$ is any homogeneous polynomial $S_\fraction$ such that
$S_\mathcal F\not\in I(\mathcal C_{\fraction})$ and $S_{\fraction} \in I(\mathcal C_{\mathcal D\setminus {\fraction}
})$.
 \item [B2)]The \emph{separator function} of ${\fraction}\subset \mathcal D$ is $SF_{\fraction}=
 \displaystyle \frac{S_{\mathcal F}}{\left( \sum_{i=1}^m x_i\right)^{s_{\fraction}}} = \sum_{\zeta \in
    \fraction} S_{\zeta}  $ where $s_{\fraction}$ is the degree of $S_{\fraction}$.
\end{itemize}
\end{definition}

\begin{example} \label{simpleMixture1.2}
The cone generated by $\fraction_{PF}=\{(0,0,1),(0,1,0),(1,0,0),(1/3,1/3,1/3)\}\subset \fraction_{P}$ is the same as the
cone generated by $\fraction=\{(0,0,1),(0,1,0),(1,0,0),(1,1,1)\} \subset \design$ where $\design\setminus\fraction =\{
(1,1,0), (1,0,1), (0,1,1)\}$. We have
    \begin{equation*}
 SF_\fraction=\displaystyle  \frac{x_1^6-2x_1x_2^5+732x_1x_2x_3^4-2x_1x_3^5+x_2^6-2x_2x_3^5+x_3^6}{(x_1+x_2+x_3)^6} \ ,
    \end{equation*}
indeed $S_\fraction(\zeta)=1$ if $\zeta\in \fraction$ or if $\zeta\in \fraction_{PF}$ and $S_\fraction(\zeta)=0$ if
$\zeta\in \mathcal D\setminus \fraction$ or $\zeta\in \mathcal \fraction_P\setminus \fraction_{FP}$. In Example
\ref{ex:12}, we  shall compute a lower degree separator for the same fraction. This shows that there exist different
ways of writing in polynomial form the separators for the same fraction. \hfill $\square$
\end{example}

We are now ready to name the two polynomial representations of a design.

\begin{definition} Let $\succ$ be a term ordering on $k$, $\fraction\subseteq \mathcal D \subset k^m$ two designs.
  \begin{enumerate}
  \item The \emph{Gr\"obner representation} of $\fraction$ with respect to $\succ$ is the reduced Gr\"obner basis of
  $I(\fraction)$ with respect to $\succ$.
  \item Let $\{d_1,\ldots,d_p\} \subset R$ be the reduced $\succ$-Gr\"obner basis of $I(\mathcal D)$ and $F$
  the indicator function of $\fraction$ in $\mathcal D$. The \emph{indicator representation} of $\fraction\subset \mathcal D$
  with respect to $\succ$ is $\{d_1,\ldots,d_p,F-1\}$.
  \end{enumerate}

Suppose now that $\fraction$ and $\mathcal D$ are designs for mixture experiments.
\begin{enumerate}
  \item The \emph{homogeneous Gr\"obner representation} of $\fraction$ with respect to $\succ$ is the reduced Gr\"obner basis of $I(\mathcal C_\fraction)$ with
    respect to $\succ$.
  \item Let $\{d_1,\ldots,d_p\} \subset R$ be the reduced  $\succ$-Gr\"obner basis of $I(\mathcal C_\mathcal D)$ and
    $S_{\mathcal D \setminus \fraction} $ the separator of $\mathcal D  \setminus \fraction$ in $\mathcal D$. The
    \emph{homogeneous indicator representation} of $\fraction\subset \mathcal D$
    with respect to $\succ$ is $\{d_1,\ldots,d_p, S_{\mathcal D \setminus \fraction} \}$.
  \end{enumerate}
\end{definition}

As a mixture design is in particular a design, it admits both the homogenous representation and the non-homogeneous
representation. Of course, when using the non-homogeneous one we loose the advantages introduced with the design cone.

While in the non mixture case $ \{ d_1, \dots, d_p, F-1 \} $ is a generating set of $ I(\fraction)$, in the mixture
case the ideal $ I(\mathcal C_\fraction) $ is the saturation of the ideal $ \langle d_1,\ldots,d_p, S_{\mathcal D
\setminus \fraction} \rangle$. The saturation $ I^{sat} $ of a homogeneous ideal $ I \subset R $ contains all the
homogeneous polynomials $ f $ such that $ f x_i^{m_i} \in I $ for some $ m_i \in \mathbb Z_{\geq 0} $ and every $ i =
1,\dots, m.$ In fact, the given generators are homogeneous and so they span only homogeneous polynomials of degree not
smaller than the degrees of the generators, while in the saturation we obtain also polynomials of degree smaller than
the degree of the generators. For example, the ideals $ \langle x \rangle $ and $ \langle x^2, xy \rangle $ in $ R =
k[x, y]$ have the cone over $ P=(0, 1) $ as zero set. Furthermore, $ I(\mathcal C_P) = \langle x \rangle = \langle x^2,
xy \rangle^{sat}$. For more on saturation see \cite{cox|little|oshea:2005} and \cite{kreuzer|robbiano:2000}.

We could substitute the requirement of reduced Gr\"obner bases with that of generating sets. Uniqueness of
representations will be lost, while there will be no longer dependence on a term-ordering.

\section{Changing representation} \label{sec:changing}

Let $F$ be the indicator function of $\fraction$ in $\mathcal D\subset k^m$, $I(\mathcal D)=\langle
d_1,\ldots,d_p\rangle$ and $I(\fraction)=\langle d_1,\ldots,d_p,g_1,\ldots,g_q \rangle$. Note that usually the
generator set $\{d_1,\ldots,d_p\}$ is known and has an easy structure, often being $\mathcal D$ a full factorial design
and hence $d_j$ a polynomial in $x_j$ for $j=1,\ldots,m$, or a simplex lattice design in the mixture case. The
difficulty and interest are related to $\fraction$. Then,
\begin{enumerate}
\item
 $I(\fraction)=\langle d_1,\ldots,d_p,F-1\rangle$. This means that once $F$ is known, a Gr\"obner basis of
 $I(\fraction)$ is  obtained by applying the Buchberger algorithm to $\{d_1,\ldots,d_p,F-1\}$.
 \item
 Vice versa, the lexicographic Gr\"obner basis (see Example \ref{ex:lex}) for $h\succ f \succ x$  of the ideal
\begin{equation*}\langle d_1,\ldots,d_p,(1-f)-\sum_j h_j g_j, fg_1,\ldots,fg_q \rangle
\end{equation*}
contains a unique polynomial of the form $f-p(x)$ where $p$ is a polynomial in the $x$ indeterminates only. Then the
evaluation function $F:\mathcal D \longrightarrow \{0,1\}$   defined as $F(d)=p(d)$ for $d\in \mathcal D$, is the
indicator function of $\fraction$ in $\mathcal D$.  See \cite[Ch. 6, Th. 2 and 3]{pistone|riccomagno|rogantin:2007},
\end{enumerate}

Items 1. and 2. above provide algorithms to switch from indicator function representation to Gr\"obner basis
representation and vice versa. While the passage from the indicator function to the Gr\"obner representation is
relatively easy as it consists of the union of polynomials, equivalently a sum of ideals, the computation of the
lexicographic Gr\"obner basis required for the passage to the indicator function representation can be computationally
expensive and often the computation does not terminate. In Section \ref{sec:algorithm} we describe a faster algorithm
for this.

Items 1. and 2. are easily adapted to the mixture/homogeneous case by considering the cone ideal and the (rational)
separator function, i.e. ideals generated by homogenous polynomials, for example, $ F - 1 $ has to be substituted with
$ S_{\fraction}-(\sum {x_i})^s$ where $s = \deg(S_{\fraction}) $. Moreover, each time we define an ideal, we must
saturate it, to compute generators of small degree.  See Item 2 of the Appendix for a {\cocoa} algorithm.

\section{An efficient algorithm}\label{sec:algorithm}

In this section, we present a different and more efficient algorithm to switch from the Gr\"obner representation to the
indicator one for a fraction $\fraction$ of a design $\design$.
The algorithm is based on the following remark: a polynomial $ f $ interpolating the indicator function $ F $ of $
\fraction $ belongs to $ I(\mathcal D \setminus \fraction) $ because of the definition of design ideal. Moreover, $ 1 - f
$ belongs to $ I(\fraction) $ for the same argument. If $ G = \{ g_1,\ldots, g_q \} $ is a Gr\"obner basis of $I(\fraction) $ then the second remark says that $ 1 - f = \sum_{i=1}^q h_i g_i,$ for some $ h_1, \dots, h_q \in R.$
Hence, $ f = 1 - \sum_{i=1}^q h_i g_i $ has normal form $ 0 $ in $ R/I(\mathcal D \setminus \fraction).$

The problem now is to efficiently choose $ h_1, \dots, h_q $.   If $ \mathcal D $ has cardinality $ N $ and $ \fraction
$ has cardinality $ n,$  then we want $ h_1, \dots, h_q $ to depend by $ N - n $ parameters because $ f = 0 $ in $
R/I(\mathcal D \setminus \fraction) $ and the dimension of $R/I(\mathcal D \setminus \fraction) $ as vector space is $N
- n $, written as $ \dim R/I(\mathcal D \setminus \fraction) = N - n $. Moreover, we would like to compute $ h_1,
\dots, h_q $ with linear algebra techniques, because they usually have smaller computational complexity than Gr\"obner
bases based algorithms.

Lemma \ref{lemma} and Theorem \ref{basis-qring} describe how to chose efficiently $ h_1, \dots, h_q.$

\begin{lemma} \label{lemma} Let $ \est_\design$ (resp. $ \est_{\fraction} $) be a monomial basis of $ R/I(\mathcal D) $ (resp. $ R/I(\fraction) $) computed by using the previously described procedure Gbasis/LT. Then $ \est_\fraction \subset \est_\design.$
\end{lemma}
\begin{proof}
$ \fraction $ is a fraction of $ \mathcal D,$ i.e. $ \fraction \subset \mathcal D$  and $I(\mathcal D)\subset
I(\fraction)$. We prove the equivalent statement: if $ x^\alpha \notin  \est_\design $ then $ x^\alpha \notin
\est_\fraction $. Let $ x^{\alpha} \notin \est_\design$ then there exists a polynomial in $ I(\mathcal D) $ whose
leading term is $ x^\alpha.$ The inclusion between the ideals shows that $ x^\alpha \notin \est_\fraction,$ and the
statement holds.
\end{proof}

Let $ x^{\alpha_1}, \dots, x^{\alpha_{N-n}} $ be the monomials in $ \est_\design\setminus \est_\fraction.$ For each $ j
= 1, \dots, N-n $ there exists $ g_{i(j)} $ in the Gr\"obner basis $ G $ of $ I(\fraction) $ such that $ x^{\alpha_j} =
m_j \operatorname{LT}(g_{i(j)}) $ for some monomial $ m_j,$ because of the construction of the monomial basis of a
quotient ring.
\begin{theorem} \label{basis-qring} With the notation as above,
 the classes of $ m_1 g_{i(1)}, \dots, m_{N-n} g_{i(N-n)} $ are a basis of $ R/I(\mathcal D \setminus \fraction).$
\end{theorem}
\begin{proof}
It is sufficient to prove that they are linearly independent. Let $ a_1, \dots, a_{N-n} \in k $ be elements of
the ground field such that
 \begin{equation*}
 a_1 m_1 g_{i(1)} + \dots + a_{N-n} m_{N-n} g_{i(N-n)}
 \end{equation*}
is the zero class in $ R/I(\mathcal D \setminus \fraction),$ that is to say, it belongs to $ I(\mathcal D \setminus
\fraction) $ because of the definition of the quotient ring. Hence, we have that $ a_1 m_1 g_{i(1)} + \dots + a_{N-n}
m_{N-n} g_{i(N-n)} \in I(\mathcal D) $ because it vanishes also on the points in $ \fraction $ being a combination of
elements in a Gr\"obner basis of $ I(\fraction).$ By construction, the leading terms of $ m_j g_{i(j)} $ are all
different and so the leading term of $ a_1 m_1 g_{i(1)} + \dots + a_{N-n} m_{N-n} g_{i(N-n)},$ say $ m_1
\operatorname{LT}(g_{i(1)},$ is in $\est_\design$. This is not possible, and so $ a_1 = 0.$ By iterating the argument a
finite number of times, we obtain that $ a_1 = \dots = a_{N-n} = 0 $ and the claim follows.
\end{proof}

We know that $ f = 1 - \sum_{i=1}^q h_i g_i $ and so we have that the indicator polynomial $ f $ can have the form $ f
= 1 - a_1 m_1 g_{i(1)} + \dots + a_{N-n} m_{N-n} g_{i(N-n)} \in I(\mathcal D \setminus \fraction).$ Hence, if we
compute the normal form $ NF(f) $ of $ f $ in the ring $ R/I(\mathcal D \setminus \fraction),$ it must be $ 0.$ In
general, $ NF(f) $ is a polynomial with monomials in $ \est_{\design \setminus \fraction} $ and linear combinations of
$ a_1, \dots, a_{N-n} $ as coefficients. Therefore, we obtain a linear system in the unknowns $ a_1, \dots, a_{N-n} $
which has a unique solution by Theorem \ref{basis-qring}. The resulting algorithm is implemented in Maple in  Item 3(a)
of the Appendix.

A few modifications are needed to adapt the algorithm to mixture designs. First, we consider homogeneous polynomials of a fixed
degree. To speed up computations, we work with polynomials of degree $ s $ where
\begin{equation*}
 s = \min \left\{t \in \mathbb Z_{> 0} \vert \dim_{\mathbb R} \left( \frac{R}{I(\mathcal D)} \right)_t = N \right\} \ .
\end{equation*}
The integer $ s $ can be easily computed by using the Hilbert function of $ R/I(\mathcal D) $ that calculates the
dimension as vector space of the degree $ t $ polynomials in $ R/I(\mathcal D) $ for every $ t \in \mathbb Z_{> 0}.$
Second, we compute monomial bases $ \est_{\design,s}$ and $ \est_{\fraction,s}$ of the degree $s$ pieces of the quotient
rings $ (R/I(\design))_s$ and $ (R/I(\fraction))_s,$ respectively. Third, the indicator function is now a ratio $ F =
f/(x_1 + \dots + x_m)^s $ where $ \deg(f) = s $ with the constraints $ f \in I(\design \setminus \fraction) $ and $ (x_1
+ \dots + x_m)^s - f \in I(\fraction).$ Hence, the changes are straightforward and the result follows also in this case.
See Item 3(b) of the Appendix for a Maple algorithm .

\begin{example}[Continuation of Example~\ref{running_example4}] \label{ex:12}
The set $\{x_1^3,x_2^3,x_2 x_3^2, x_3^3 \}\subset \est_{\fraction_P} $ is identified by the fraction $\mathcal F = \{
(1,0,0),(0,1,0),(0,0,1),(1,1,1) \} \subset \fraction_P $. The ideal of the fraction is $\langle x_3 (x_1-x_2),(x_1-x_3)
x_2,x_2 x_3 (x_2-x_3) \rangle$, while the ideal of $\fraction_P \setminus \mathcal F $ is $\langle x_3 (x_1+x_2-x_3),
x_2 (x_1-x_2+x_3), x_1^2-x_2^2+2 x_2 x_3 - x_3^2, x_2 x_3 (x_2-x_3) \rangle$. Then, by applying the previous algorithm,
we obtain the following indicator function
\begin{equation*}  SF_\fraction =
\frac{x_1^3 + 3x_1^2 x_2 + 3x_1^2 x_3 - 5x_1 x_2^2 + 30x_1 x_2 x_3 - 5x_1 x_3^2+ x_2^3 + 11x_2^2 x_3 - 13x_2 x_3^2 +
x_3^3}{(x_1+x_2+x_3)^3}.
\end{equation*}
The following Maple script performs the computation using the algorithm in Item 3(b) of the Appendix.
{\scriptsize
\begin{verbatim}
 var:= [x,y,z] -- we change x_1, x_2, x_3 to x, y, z, respectively
 EstX:= {x^3, x y^2,y^3,x*y*z,x*z^2,y*z^2,z^3}  -- monomial basis of (R/I(F_P))_3
 EstY:= {x^3,y^3,y*z^2,z^3}  --  monomial basis of (R/I(F))_3
 GY:= {z*(x-y),(x-z)*y,y*z*(y-z)}
 GXMinusY:= z*(x+y-z), y*(x-y+z), x^2-y^2+2*y*z-z^2, y*z*(y-z)

 G_to_F_homo(GY,GXMinusY,EstX,EstY,var,t)
 x^3 + 3x^2y + 3zx^2 - 5xy^2 + 30xyz - 5xz^2+ y^3 + 11zy^2 - 13yz^2 + z^3
\end{verbatim}}
\hfill $\square$
\end{example}

\section{Example}\label{sec:ex}

The simplex centroid design is defined in \cite{scheffe:63} and used for experiments with mixtures. In $k$ factors it has
$2^k-1$ points. Fractions of the simple centroid design with many less points are defined in
\cite{mcmonkey|mezey|dixon|grenberg:2000} and used to screen for significant factors.  Their definition depends on an
integer parameter $p$ and the double interactions are completely aliased over any such fraction in sets of size $p$. A typical example is $\fraction_{MC}$ below
\begin{equation*} 
\begin{split}
\fraction_{MC}=\{
 &(1,0,0,0,0,0,0,0,0),(0,1,0,0,0,0,0,0,0),(0,0,1,0,0,0,0,0,0), \\
 &(0,0,0,1,0,0,0,0,0),(0,0,0,0,1,0,0,0,0),(0,0,0,0,0,1,0,0,0), \\
 &(0,0,0,0,0,0,1,0,0),(0,0,0,0,0,0,0,1,0),(0,0,0,0,0,0,0,0,1), \\
 &(1/3,1/3,1/3,0,0,0,0,0,0),(1/3,0,0,1/3,0,0,0,1/3,0),(0,1/3,0,0,1/3,0,0,0,1/3),\\
 &(0,0,1/3,0,0,1/3,1/3,0,0),(0,0,0,1/3,1/3,1/3,0,0,0),(0,1/3,0,1/3,0,0,1/3,0,0),\\
 &(0,0,1/3,0,1/3,0,0,1/3,0),(1/3,0,0,0,0,1/3,0,0,1/3), (0,0,0,0,0,0,1/3,1/3,1/3),\\
 &(1/3,0,0,0,1/3,0,1/3,0,0),(0,1/3,0,0,0,1/3,0,1/3,0),(0,0,1/3,1/3,0,0,0,0,1/3) \} \ .
\end{split}   \end{equation*}

The design $\fraction_{MC}$ can be seen as a subset of the simplex centroid design in $9$ factors  $\mathcal D_1$ and
also as a fraction of the simplex centroid which  includes the corner points of the simplex and the points with
coordinates equal to zero or 1/3, call it $\mathcal D_2$.

$\texttt{IdealOfPoints}(\fraction_{MC})$ with respect to the default term ordering in {\cocoa} is generated by 43
polynomials, while $\texttt{IdealOfProjectivePoints}(\fraction_{MC})$ is generated by 42 ones.  See
\cite{maruri-aguilar|notari|riccomagno:2007} for a discussion of these results.

The indicator functions of $\fraction_{MC}$ in $\mathcal D_1$ and of $ \fraction_{MC} $ in $\design_2 $, and the
separator of $\fraction_{MC}$ in $\mathcal D_2$ have been computed by using the Maple algorithm in the Appendix in less
than 10 sec. the first one, and in less than 1 sec. the last two. The indicator functions of $\fraction_{MC}$ in
$\mathcal D_1$ is a combination of 444 terms, in $\design_2 $ is a combination of 70 terms, and the separator of
$\fraction_{MC}$ in $\mathcal D_2$  is a combination of 165 terms.

\section{Discussion}

This note regards two polynomial representations of an experimental design $\fraction$ one of which uses the indicator
function of $\fraction$ in $\mathcal D$ and the other one uses Gr\"obner bases which does not require to think of
$\fraction$ as a fraction of the larger design $\mathcal D$. The Gr\"obner representation depends on a technical
object: a term ordering, while the indicator function is most informative with a complex coding of the factor levels.
In applied work term orderings have been used to the advantage of the statistical analysis, see
\cite{holliday|riccomagno|wynn|pistone:99}. In \cite{pistone|rogantin:2007} it is shown that the real part of the
complex response retains most of the properties of the full complex response while having a clearer physical
interpretation. Moreover, notice that most of the properties discussed in Section 2.1 depend intrinsically on the level
coding. A trivial example is that the $2^2$ full factorial design with levels $\pm 1$ is orthogonal for
$\{1,x_1,x_2,x_1x_2\}$ while with levels $\{0,1\}$ it is not.

Both representations can be used to identify alias relations imposed by $\fraction$ on $\est_\mathcal D$, which is a
finite set, or on some other sets of monomials, possible all the infinite set of monomials. Furthermore, both provide a
vector-space basis of the response space. This is hierarchical for the Gr\"obner basis representation.

The choice as to which representation to use should be made in the light of the interests of the practitioner. If the
responses have been collected and standard or slower techniques have not returned a satisfactory statistical analysis
or are not implementable (maybe because there are missing values with respect to the planned experiment), then it seems
convenient to apply the GBasis/LT procedure. This returns an identifiable  hierarchical regression model and the alias
relations in the Gr\"obner basis can be used to change model terms with more significant or interpretable model
interactions. Instead, prior to data collection, the indicator function seems a useful tool to select a design with
relevant properties by working on the coefficients of the indicator functions. Two issues have to be considered: 1.
the need of a complex coding for some properties and 2. the need to solve a system of polynomial equations to obtain
the point coordinates. Point 1 has been discussed above. In most practical cases Point 2 can be easily addressed by
computing a Gr\"obner basis with respect to a lexicographic ordering. Generally the joint use of the two representation
seems advisable, also in the light of the switching algorithms in Sections 3 and 4. Indeed, if one representation is
known then the other one can be computed using techniques of linear algebra. This is possible because designs are zero
dimensional varieties. The complexity of the algorithms in the Appendix is essentially the complexity of the computation of the normal form of a polynomial w.r.t. a Gr\"obner basis. In fact the last step of the algorithm consists in solving a linear system which has a smaller computational complexity. For the large designs of Section \ref{sec:ex} the algorithm gave the solutions in just a few seconds, as previously mentioned.

Finally, we wanted to have in the public domain a complete set of computer functions to perform the computations in the
diagram of the Introduction.

\section*{Appendix}\label{appendix}
\begin{enumerate}
\item The \cocoa\ code for the indicator function of $\fraction_{A}$ in Example~\ref{running_example2}.

{\scriptsize
\begin{verbatim}
Use S::=Q[x[1..2]];
 Define InFu(Points,D); ND:=Len(D); PA:=NewList(ND,0); P:=NewList(ND);
 For H:=1 To Len(Points) Do
  For K:=1 To Len(PA) Do If Points[H]=D[K] Then PA[K]:=1  End; End;
 End;
 IdD:=IdealAndSeparatorsOfPoints(D);
  For K:=1 To Len(PA) Do  P[K]:=PA[K]*IdD.Separators[K]  End;
 F:=Sum(P); Return F;
 End;
 D:=Tuples([-1,0,1], NumIndets()); PointsF:= [[1,0],[-1,0],[0,1],[0,-1]];
InFu(Points,D);
 \end{verbatim}}

\item The {\cocoa} code for   $S_\fraction$ of Example \ref{simpleMixture1.2},
 with the algorithm described in Item 2 of Section \ref{sec:changing}.

 {\scriptsize
\begin{verbatim}
 Use T::=Q[f h x[1..3]], Lex;
 Set Indentation;
 D:=[x[2]^2x[3] - x[2]x[3]^2, x[1]^2x[3] - x[1]x[3]^2, x[1]^2x[2] - x[1]x[2]^2];
      -- simplex lattice
 G:=[x[1]x[3] - x[2]x[3], -x[1]x[2] + x[2]x[3]]; -- (0,0,1),(0,1,0),(1,0,0),(1,1,1)
 P:=(x[1]+x[2]+x[3])^3; S:=[];
 For I:=1 To Len(G) Do
    L:=ConcatLists([D,[P-f-hG[I],fG[I]]]);
    Id:=Saturation(Ideal(L),Ideal(x[1],x[2],x[3]));
    GB:=ReducedGBasis(Id);     S:=Concat( [f-GB[1]], S);
 EndFor;
 SF:=NF(Product(S),Ideal(D)); SF;
\end{verbatim}}

\item The  Maple code of the procedure described in Section \ref{sec:algorithm} to compute the indicator functions of
$Y\subset X$ where $Y$ and $X$ are sets of points. The affine and the projective cases are considered. Notice that all
computations are with respect to the tdeg term ordering of var. This can be changed by the user.

\begin{enumerate}
\item  \emph{Affine case.} In input the procedure requires:\\
\begin{tabular}{rcl}
        GY &=& Gr\"obner basis of $I(Y)$, \\
        GXMinusY &=&  Gr\"obner basis of $I(X \setminus Y)$, \\
        EstX &=& standard basis of $R/I(X)$, \\
        EstY &=& standard basis of $R/I(Y)$, \\
        var &=& list of indeterminates
        \end{tabular}

The Output is the polynomial representation of the indicator function in $R/I_X$.\\
{\scriptsize
 \begin{verbatim}
 G_to_F := proc(GY, GXMinusY, EstX, EstY, var)
 local E, InY, ly, L, flag, tt, F, i, m, l, v, CC, VV;
 with(Groebner);
  v := op(var);
  E :=`minus`(EstX,EstY);
  InY := [seq(leadterm(GY[i],tdeg(v)),i = 1 ..nops(GY))];
  ly := nops(InY);
  L := [];
  for m in E while true do
   flag :=0;
   while flag = 0 do for i to ly do if gcd(InY[i],m) = InY[i]
     then flag := 1; tt := i end if end do end do;
   L := [op(L),GY[tt]*m/leadterm(GY[tt],tdeg(v))]
  end do;
  F := 0;     i := 0;
  for l in L while true do i := i+1; F := F+a[i]*l end do;
  CC := coeffs(normalf(1-F,GXMinusY,tdeg(v)),[v]);
  VV := solve(\{CC\});
  expand(subs(VV,F+1))
 end proc;
 \end{verbatim}}

 \item \emph{Projective case.} In input the procedure requires:\\
\begin{tabular}{rcl}
        GY &=& Gr\"obner basis of $I(Y)$, which is an homogeneous ideal, \\
        GXMinusY &=&  Gr\"obner basis of $I(X \setminus Y)$,  which is an homogeneous ideal, \\
        EstX &=& standard basis of degree $t$ of $R/I(X)$\\
        EstY &=& standard basis of degree $t$ of $R/I(Y)$, \\
        var &=& list of indeterminates,
        \end{tabular}
\\
\noindent here $t$ is the minimal degree for which $(R/I_X)_t$ as a vector space has the same dimension as the number of
points in $X$.

The Output is the numerator of the ratio of polynomials giving the separator function of $Y$ in $X$. Its denominator is
the $t$-power of the sum of the variables.

 {\scriptsize
\begin{verbatim}
 G_to_F_homo:=proc(GY,GXMinusY,EstX,EstY,var,t)
 local E,InY,ly,L,flag,tt,F,i,m,l,v,CC,VV,S:
 with(Groebner):
  v:=op(var):
  E:=EstX minus EstY:
  InY:=[seq(leadterm(GY[i],tdeg(v)), i=1..nops(GY))]:
  ly:=nops(InY):
  L:=[]:
  for m in E do
   flag:=0:
   while(flag = 0) do
     for i from 1 to ly do if(gcd(InY[i],m)=InY[i])then flag:=1: tt:=i: fi: end:
   end:
   L:=[op(L),GY[tt]*m/leadterm(GY[tt],tdeg(v))]:
  end:
  F:=0: i:=0:
  for l in L do  i:=i+1: F:=F+a[i]*l:  end:
  S:=sum('v[k]',k=1..nops(var)):
  CC:=coeffs(normalf(S^t-F, GXMinusY, tdeg(v)),[v]):
  VV:=solve(\{CC\}):
  expand(subs(VV,S^t-F)):
end proc;
\end{verbatim}}
\end{enumerate}

\end{enumerate}

\end{document}